\documentclass[
reprint,
superscriptaddress,
amsmath,amssymb,
aps,
pra,
]{revtex4-2}

\usepackage{soul} 
\usepackage{xcolor}
\usepackage{tabularx}
\usepackage{ulem} 
\usepackage{dcolumn}
\usepackage{bm}
\usepackage{bbm}
\usepackage{overpic}
\usepackage{color}
\usepackage{psfrag,graphicx}
\usepackage{verbatim}
\usepackage{enumitem}
\usepackage{dsfont}

\usepackage{float}
\usepackage{mathtools}
\usepackage[english]{babel}
\addto\captionsenglish{}

\definecolor{mygrey}{gray}{0.35}
\definecolor{myblue}{rgb}{0.2,0.2,0.8}
\definecolor{myzard}{cmyk}{0,0,0.05,0}
\definecolor{mywhite}{rgb}{1,1,1}
\definecolor{myred}{rgb}{0.9,0.1,0.}
\usepackage[colorlinks=true,citecolor=myblue,linkcolor=myblue,urlcolor=myblue]{hyperref}
\newcommand{\defeq}{\vcentcolon=}

\usepackage[makeroom]{cancel}

\newenvironment{proof-of}[1]{\medskip\noindent\textbf{Proof of {#1}.}}{\hfill$\blacksquare$\medskip}
\newcommand{\ket}[1]{\left\vert#1\right\rangle}
\newcommand{\bra}[1]{\left\langle#1\right\vert}
\newcommand{\braket}[2]{\ensuremath{\langle #1 | #2 \rangle}}

\usepackage{tcolorbox}

\usepackage{etoolbox} 
\makeatletter
\def\split@frontmatter@footnote#1{
  \begingroup
    \def\@join##1##2##3{%
      \def\@tempa{##2}%
      \@ifx{\@tempa\@empty}{}{##2\comma@space}%
      \frontmatter@footnote{##3}}%
    #1%
  \endgroup}
\patchcmd{\doauthor}
  {\frontmatter@footnote{#2}}
  {\split@frontmatter@footnote{#2}}
  {}
  {\PackageWarningNoLine{preamble}{Could not patch \string\doauthor;
    author notes will stay fused (native REVTeX behaviour)}}
\makeatother

\begin{document}

\title{Entanglement-assisted Hamiltonian dynamics learning}

\author{Ayaka Usui}
\thanks{These authors contributed equally to this work.}
\email{ayaka.usui@uab.cat}
\affiliation{Física Teòrica: Informació i Fenòmens Quàntics, Departament de Física, Universitat Autònoma de Barcelona, 08193 Bellaterra, Spain}

\author{Guillermo Abad-L\'{o}pez}
\thanks{These authors contributed equally to this work.}
\email{guillermo.abad@uab.cat}
\affiliation{Física Teòrica: Informació i Fenòmens Quàntics, Departament de Física, Universitat Autònoma de Barcelona, 08193 Bellaterra, Spain}
\affiliation{Qilimanjaro Quantum Tech, Carrer de Veneçuela 74, 08019, Barcelona, Spain}

\author{Hari krishnan SV}
\email{harikrishnan.sasidharan@uab.cat}
\affiliation{Física Teòrica: Informació i Fenòmens Quàntics, Departament de Física, Universitat Autònoma de Barcelona, 08193 Bellaterra, Spain}

\author{Anna Sanpera}
\email{anna.sanpera@uab.cat}
\affiliation{Física Teòrica: Informació i Fenòmens Quàntics, Departament de Física, Universitat Autònoma de Barcelona, 08193 Bellaterra, Spain}
\affiliation{ICREA, Pg.~Llu\'{i}s Companys 23, 08010 Barcelona, Spain}

\author{Some Sankar Bhattacharya}
\email{somesankar.bhattacharya@uab.cat}
\affiliation{Física Teòrica: Informació i Fenòmens Quàntics, Departament de Física, Universitat Autònoma de Barcelona, 08193 Bellaterra, Spain}

\begin{abstract}
Approximating the dynamics generated by complex many-body Hamiltonians with simpler models is a central challenge in quantum simulation. Quantum generative adversarial networks (QGANs) provide a powerful framework for learning unitary dynamics, often achieving more compact representations than standard Trotter-based approaches. However, their practical performance is frequently limited by optimization plateaus and suboptimal local minima. Here, we propose an entanglement-assisted learning strategy that mitigates these bottlenecks by dynamically coupling the system during training to a randomly initialized auxiliary system. This temporary enlargement of the Hilbert space redistributes quantum correlations and improves trainability. We numerically demonstrate its advantage in learning Hamiltonian dynamics with a hardware-efficient ansatz. By analyzing connectivity, coupling timing, and parameter initialization of the auxiliary system, we identify the key ingredients behind the observed performance gains. Our results highlight a strategy for overcoming optimization bottlenecks in variational quantum algorithms, with applications in quantum simulation, quantum control, and quantum machine learning.

\end{abstract}
\maketitle

\section{Introduction} 
The properties of quantum many-body systems are intimately linked to the operator content of the Hamiltonian that governs their dynamics \cite{olsacher2025hamiltonian}. Given incomplete data originating from an unknown Hamiltonian, quantum Hamiltonian learning (QHL) encompasses a variety of methods aimed at inferring its operator structure, i.e., the type of interactions and their strengths~\cite{Wiebe2014Hamiltonian, Wang17, Bairey19,Li2020Hamiltonian,Wang2015Hamiltonian,Bienias2021Meta,Hangleiter2024Robustly,Zubida2021Optimal,Wenjun2023Robust,Pastori2022Characterization,Zhao2025Learning,Arunachalam2025Testing,Hu2025Ansatz}. Such data may stem from experimental limitations such as restricted access to expectation values of certain observables or from theoretical constraints, for instance, having access to only a single eigenvector of the so-called parent Hamiltonian~\cite{Qi2019Determining}.
Quantum simulators, on the other hand, are purpose-built experimental platforms, either analog or digital, designed to probe and extract properties of complex quantum systems that are intractable for classical computation~\cite{Feynman1982Simulating,Lloyd1996Universal,Altman2021Quantum,Daley2022Practical,Beverland2022Assessing}. These two approaches, QHL and quantum simulation, can be viewed as dual tasks: one seeks to reconstruct a parent Hamiltonian from restricted data, while the other generates data from non-trivial Hamiltonians to infer properties of complex many-body quantum systems.

Between these two paradigms lies the task of Hamiltonian approximation, replacing a complex, possibly unknown Hamiltonian with a simpler, experimentally feasible one \cite{cubitt,Usui2024Simplifying}. This goal is relevant to both QHL and quantum simulation, as it enables tractable modeling and experimental realization of otherwise inaccessible quantum systems. A representative example involves approximating highly non-local Hamiltonians—such as those arising, for instance, in lattice gauge theories with three- or four-body interactions—by simpler Hamiltonians involving only two- and one-body interactions. More generally, given an unknown complex target Hamiltonian $H_T$ of locality $k$, one may ask: What is the best approximation of its dynamics? Specifically, what is the optimal effective Hamiltonian we can generate, $H_G$, composed only of one- and two-body interactions, that reproduces the same unitary evolution~\cite{Usui2024Simplifying}. 

Generically, the unitary evolution of non-trivial time-independent Hamiltonians, $U_T = e^{-iH_T t}$, involves computing the exponential of a sum of non-commuting operators. This is usually approached using the Trotter decomposition, which approximates the total time evolution by dividing it into very short time steps and breaking the exponential of the sum into a product of simpler exponentials.
In digital quantum computation, the Trotter approximation can be implemented as a sequence of one- and two-qubit gates~\cite{Lloyd1996Universal}. If the sequence of gates is sufficiently long, the approximation converges to the exact unitary evolution. However, the Trotter decomposition does not necessarily yield the most efficient gate sequence to simulate the dynamics. To address this limitation, various alternative techniques have been proposed, see for instance~\cite{Gong2024Complexity}.

Here, we employ quantum generative adversarial networks (QGANs), an unsupervised quantum machine learning model proposed in~\cite{Lloyd2018Quantum,Dallaire2018Quantum}, as a quantum extension of its successful classical counterpart, generative adversarial networks~\cite{goodfellow2014generative}. The learning efficiency of QGANs has been reported in multiple instances for both quantum and classical data~\cite{superconducting_circuit,multiple_supercond_circuits,q_state_tomo_gans,random_dist,pure_state_approx,Exp_QGAN_Image_gen,Qugan}. 
An example of such efficiency for learning Hamiltonian unitary evolution was provided in~\cite{Chakrabarti2019Quantum}, where for a fixed time, the unitary evolution of a \(3\)-qubit Heisenberg Hamiltonian required \(\sim 10^4\) gates using the Trotter approximation, whereas the QGAN implementation achieved a comparable average fidelity with just \(52\) gates.

Nevertheless, as in many machine learning protocols involving the minimization of a cost function in a large phase space, optimization of the training algorithm is hindered by unstable convergence and the vanishing gradient of the cost function~\cite{Ngo2023Survey}. As a consequence, the training may exhibit stagnation, even when sufficient resources are allocated in the QGAN. For example, increasing the number of gates may increase the expressivity of the algorithm and improve the probability of a successful optimization; however, redundant gates can also stagnate the training, and even if the training goes well, it results in a lengthier final decomposition. 

In this work, we address the above challenges and propose an efficient strategy for learning unitaries within the QGAN framework. Our goal is twofold. First, we aim to approximate the dynamics of unknown complex Hamiltonians with simpler representations containing only a minimal number of one‑body and two‑body terms. Second, we show that incorporating an auxiliary system during the training phase helps mitigate stagnation by leveraging entanglement to facilitate the learning process.
In doing so, our method not only enhances the overall efficiency of QGAN training but also enables more compact decompositions of complex unitary transformations. 

The remainder of this article is organized as follows. In Section~\ref{sec:unitarylearning}, we introduce the framework of QGANs for unitary learning. 
In Section~\ref{sec:aqgan}, we present the physical mechanism and the advantages of our entanglement-assisted learning strategy. 
In Section~\ref{sec:analysis}, we analyse the learning strategy by computing the dynamical behaviour of entanglement and studying the gradients of the cost function.
Finally, in Section~\ref{sec:summ}, we summarize our findings and discuss future research directions.

\section{Unitary learning via QGANs}
\label{sec:unitarylearning}

A learning process for both classical and quantum states can be formulated as an adversarial game between two players: a generator $G$ and a discriminator $D$, trying to learn an unknown state $T$. The goal of the generator is to produce states that resemble, as closely as possible, the target $T$, while the discriminator seeks to distinguish them using a cost function, $\mathcal L$, which quantifies their distance. To all effects, the target $T$ can be treated as a black box.

During the learning process, the generator and the discriminator are both updated alternately.
First, with the generator fixed, the discriminator parameters $\phi$ are improved to better distinguish the states, maximizing the cost function. Then, given the just-found, optimized discriminator, the generator parameters $\theta$ are adjusted to approach the target state, minimizing the cost function. This adversarial interaction can be naturally framed as a minimax game: $\min_{\theta} \max_{\phi} \  \mathcal{L}\left(\Psi_T,\Psi_{G(\theta)}, D(\phi)\right)$, where the improvement of one component drives the refinement of the other, allowing for unsupervised learning of states.

Assuming that the generator and the discriminator both have enough capacity, the optimization landscape can be chosen convex~\cite{Lloyd2018Quantum}, with the training eventually reaching its Nash equilibrium. There, the generator state is identical to the target, and the discriminator cannot tell them apart.
Therefore, while in general the cost function both increases and decreases during training due to the two opposing optimizations, the adversarial learning maximizes the fidelity between the generated state and the target state. 

A quantum circuit architecture is typically employed to implement QGANs (optimal control methods can also be used~\cite{Kim2024Hamiltonian}). 
The generator consists of a parametrized quantum circuit and is constructed with multiple layers of parametrized gates, as schematically depicted in Fig.~\ref{fig:qgan}(b).
The discriminator is framed as a parametrized operator, which characterises the distance between the target state and the generator state through the cost function.
As a cost function, we adopt the quantum Wasserstein semi-metric~\cite{Chakrabarti2019Quantum}, an optimal-transport cost~\cite{Learning_Quantum_data, Karol25, KarolPrivate} between the generated state and the target state with a SWAP-based cost operator. 
This formulation was shown to improve the robustness and the scalability of the training of quantum generative models~\cite{Chakrabarti2019Quantum} (see Appendix~\ref{app:wasserstein} for more details and its implementation).

\begin{figure}[t]
    \includegraphics[width=0.9\linewidth]{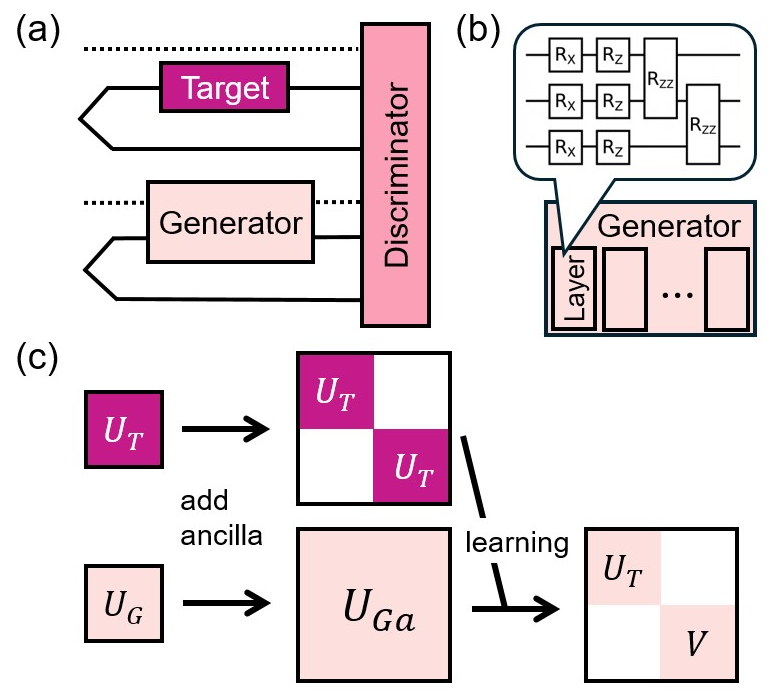}
    \caption{(a) Layout of entanglement-assisted QGANs.
    The auxiliary system (dotted line) is embedded in the generator. The input state is given by Eq.~\eqref{eq:bipartite_extra}. 
    (b) Ansatz of each layer in the generator,
    consisting of single-qubit rotations $R_X$, $R_Z$, and two-qubit interactions $R_{ZZ}$.
    (c) Matrix representation of entanglement-assisted learning, where $U_T$, $U_G$, and $U_{Ga}$ are the operators of the target, the original generator, and the auxiliary-expanded generator, respectively, with $V$ a new redundant unitary not used for training.
    The white areas indicate zeros. 
    The expanded $U_{Ga}$ is trained against a block diagonal target until the top-left generator subspace approximates the target.
    }
    \label{fig:qgan}
\end{figure}

At the cost of increasing the number of input qubits, QGANs can be designed for the learning of unitaries by exploiting Choi-Jamio{\l}kowski isomorphism \cite{Choi, Jamiol}. 
This relies on maximal bipartite entanglement ~\cite{Chakrabarti2019Quantum,Kim2024Hamiltonian} and is more efficient in terms of the number of iterations for training than a classical alternative, which sends batches of multiple different states to both the target and the generator~\cite{Exp_QGAN_Image_gen}. 
Assuming that the target acts on an $N$-qubit system, one prepares, as input, a maximally entangled 2$N$-qubit state of the form
\begin{equation} \label{eq:bipartite}
    \ket{\Omega}=
    \frac{1}{\sqrt{d}}
    \sum_{i=0}^{d-1}
    \ket{i}_{\text{A}}
    \ket{i}_{\text{B}}
    ,\end{equation}
where $d=2^N$ is the dimension of the Hilbert space, 
while $\ket{i}_{A(B)}$ denotes a state of the computational basis of $N$-qubits of system A(B).
Such a configuration is symbolically sketched by solid lines in Fig.~\ref{fig:qgan}(a).

The fidelity between the output states of the target unitary, $\mathbb I_A\otimes U_T\ket{\Omega}$, and of the generator, $\mathbb I_A\otimes U_G\ket{\Omega}$, when both act only in one half of the system (B), is given by
\begin{align} \label{eq:F_choi}
   F=&\Big|
    \bra{\Omega}
    \left(\mathbb I_A \otimes U_G
    \right)^{\dagger}
    \left(
    \mathbb I_A\otimes U_T
    \right)
    \ket{\Omega}
    \Big|^2 \nonumber \\
    =&\frac{1}{d^2}
    \left|
    \mathrm{Tr}[
    U_G^{\dagger}U_T
    ]
    \right|^2,
\end{align}
which is proportional to the squared modulus of the Hilbert-Schmidt inner product of the operators $\langle U_G U_T\rangle = \mathrm{Tr}(U_G^\dagger U_T)$. The fidelity reaches its maximum value, $F=1$, if and only if $U_G$ and $U_T$ differ at most by a global phase. Thus, training takes the generator towards $U_G\simeq U_T$, despite $H_G\neq H_T$.

\section{Entanglement-Assisted Learning}\label{sec:aqgan}

The QGAN learning process can be framed as an optimization process that aims to provide the best approximation to $U_T$ within a fixed set of resources. These involve the gates in $G$, the number of iterations, and the total time evolution. Typically, the central question is determining which generator ansatz and how many layers are required to achieve a sufficiently close approximation to the target, given a fixed number of training iterations. In practice, however, it is not always possible to find the optimal decomposition of $U_T$  within a reasonable training effort (number of iterations), due to QGANs facing several variational circuit pathologies, like most quantum machine learning protocols do~\cite{McClean18,QVA_traps}. Due to restricted capacities, redundancies, and finite training times, optimization landscapes exhibit learning plateaus or local minima, creating unfavorable initial points for which the fidelity saturates well below one or does not even converge to any given value.

To address this, we increase the generator phase space by adding an \textit{interacting auxiliary qubit}. This results in an improved expressivity for the generator, leading to a better convergence towards the gate decomposition of $U_T$ as we shall see below. Our proposal for learning $U_T$ with an auxiliary-assisted generator modifies the input state for the generator and the target both to the product state of the Choi maximally entangled state of $2N$-qubits with an auxiliary qubit: $\ket{\Omega'}= \ket{\Omega}\otimes \ket{0}_a$, where for simplicity we initialize the auxiliary qubit in the state $\ket{0}$ and use the subscript $a$ to denote it, and the input state is given by
\begin{align} \label{eq:bipartite_extra}
    \ket{\Omega'}
    &=
    \frac{1}{\sqrt{d}}
    \left (\sum_{i=0}^{d-1}
    \ket{i}_{\text{A}}
    \ket{i}_{\text{B}}\right)\otimes \ket{0}_a
    .
    \end{align}
As shown in Fig.~\ref{fig:qgan}(a), this modification is only relevant to the generator, but since the discriminator requires comparing the outputs of our black box target $T$ and the generator $G$, the Hilbert space at both ends must coincide, so a decoupled auxiliary qubit is also passed together with the target output.

The fidelity between the target and the generator output states is now given by
\begin{align}
    F
    &=
    \Big|\langle \Omega'|\left(\mathbb I_A \otimes U_{Ga}^\dagger \right)
    \left( \mathbb I_A \otimes U_T \otimes \mathbb I_a \right) |\Omega'\rangle \Big|^2 \nonumber\\
    &=
    \frac{1}{d^2}
    \left|\mathrm{Tr}[{U'_G}^{\dagger}U_{T}]\right|^2,
    \label{eq:Choi_ancillary}
\end{align}
where $U_{Ga}$ is the unitary operator acting on half of the maximally entangled state plus the auxiliary qubit, as depicted in Fig.~\ref{fig:qgan}, and ${U'_G}\equiv \;_a\langle 0|{U_{Ga}}|{0}\rangle_a$ is its projection onto the auxiliary subspace where the auxiliary qubit remains unchanged. 
As previously, high values of the fidelity correspond, up to a global phase, to ${U'_G}\simeq U_T$. 
Therefore, even though on each layer of the generator the auxiliary qubit interacts with the other qubits, the learning process should lead $U_{Ga}$ to a block diagonal form in the auxiliary subspace, with $U_T$ in one of the blocks, as shown schematically in Fig.~\ref{fig:qgan}(c). 

Next, we demonstrate how the learning process is enhanced through the auxiliary-assisted training. Specifically, we discuss in detail why the introduction of the auxiliary system during training can help to overcome stagnation. 
Depending on the entanglement distribution generated by the auxiliary system connectivity, it can help training to escape local minima and plateaus.

As a proof of concept, we consider the unit-time evolution operator $U_T=e^{-iH_T}$ acting on three qubits for two representative target Hamiltonians: a simple three-local qubit Hamiltonian, $H_T=\sigma_z^1\sigma_z^2\sigma_z^3$, and a Hamiltonian containing non-commuting terms, $H_T=\sigma_{z}^1\sigma_{z}^2\sigma_{z}^3+\sigma_{x}^1\sigma_{x}^2\sigma_{x}^3/4$.
Nevertheless, we have benchmarked our method on several additional Hamiltonians. Furthermore, to address the scalability of the approach, we also consider a four-qubit system in Appendix~\ref{sec:bigger_H}.

As already pointed out, the generator consists of multiple layers, and each layer adopts an ansatz composed of nearest-neighbour interactions and local rotations, as illustrated in Fig.~\ref{fig:qgan}(b). 
For instance, systems of superconducting qubits do not have as much flexibility in terms of the qubit connectivity as other systems such as trapped ions, but this ansatz is easily implemented in any type of hardware.
The discriminator consists of generalized local measurements, parametrized by a linear combination of Pauli operators, before a local computational basis measurement in each qubit. The discriminator and generator parameters are both randomly initialized at the start of each training, and updated alternately via classical momentum-based gradient descent (see Appendix~\ref{app:wasserstein} for the explicit parameter update equations).
This generator architecture becomes universal in the limit of infinitely many layers, as shown in Ref.~\cite{Barenco1995Elementary}. Consequently, the expressivity of our generator increases with the number of layers. 
Based on our numerical results, the shortest gate sequence that performed sufficiently well is obtained using three layers. We fix this configuration for the remainder of our study.

\begin{figure}
\centering
    \ \ \includegraphics[width=0.93\linewidth]{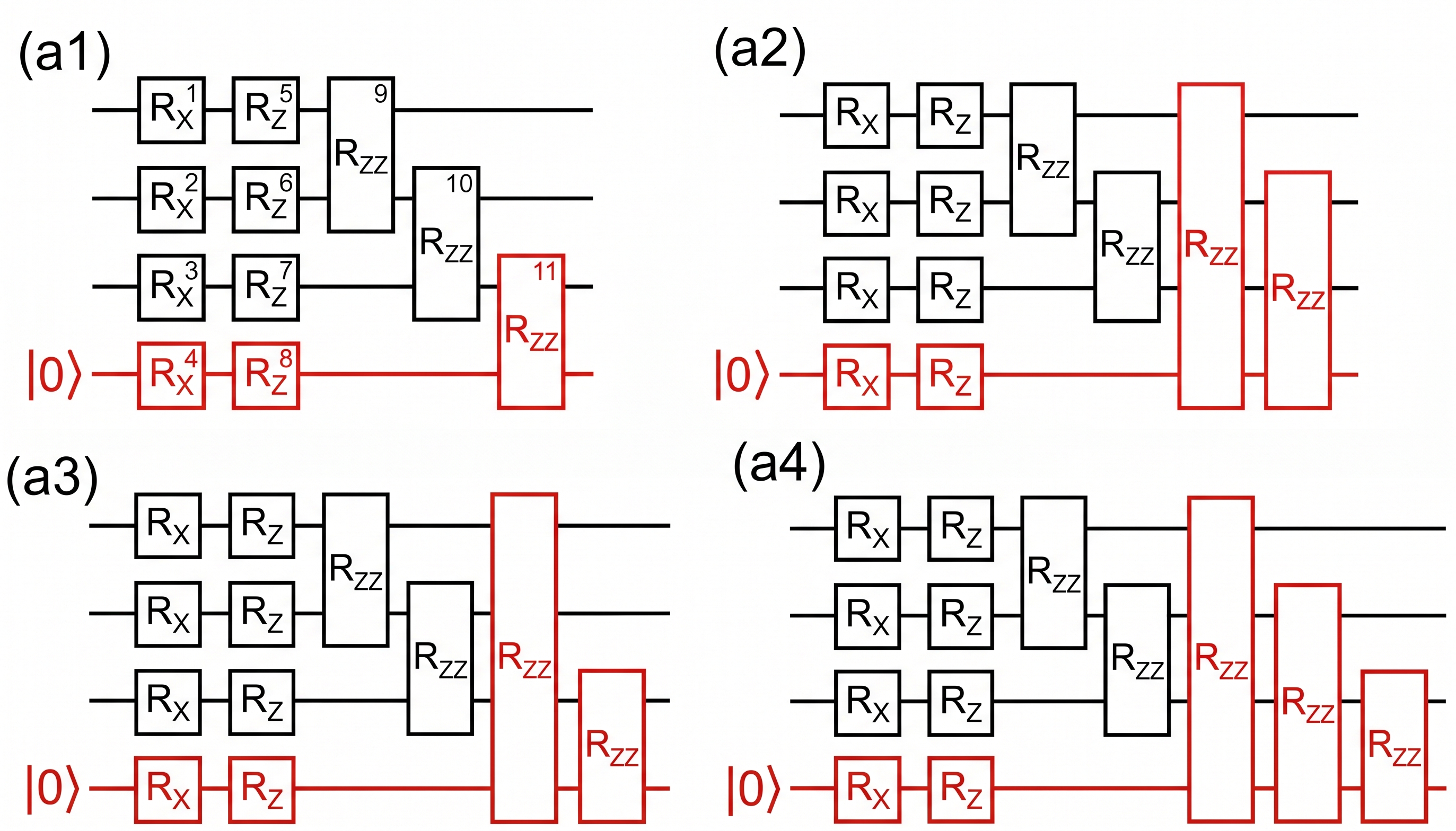} \\
        \vspace{5px}
        \includegraphics[width=0.9\linewidth]{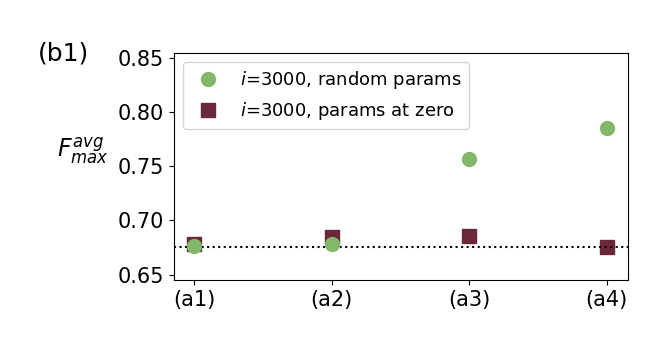} \ \ \ \ \ \\
        \vspace{-10px}
        \includegraphics[width=0.9\linewidth]{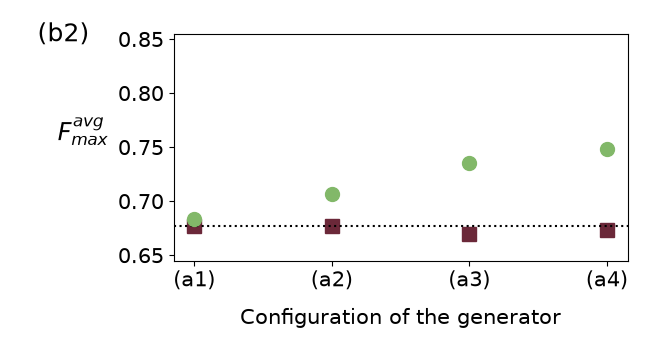} \ \ \ \ \
     \caption{
     Entanglement-assisted unitary learning via an auxiliary qubit. 
     (a1-4) Auxiliary-qubit configurations used in the training for each layer, consisting of single-qubit rotations $R_X$, $R_Z$ and two-qubit interaction gates $R_{ZZ}$.   
     The parameters of the gates are indexed.
    (b) Average maximum fidelities over 100 cases, $F^{\mathrm{avg}}_{\max}$, for the different configurations.   The green dots (brown squares) denote when the parameters of the auxiliary qubit are initialized randomly (set to zero).
     The dotted horizontal line indicates the reference fidelity achieved in the absence of the auxiliary qubit.
     In (b1,b2), we target $H_T=\sigma_z^1\sigma_z^2\sigma_z^3$ and $H_T=\sigma_z^1\sigma_z^2\sigma_z^3+\sigma_x^1\sigma_x^2\sigma_x^3/4$, respectively.} 
    \label{fig:midrun}
\end{figure}

Each training sequence consisted of up to 6000 iterations (two sequential 3000-iteration stages when the auxiliary qubit was added mid-training), or until a target fidelity of $99\%$ was reached in iteration $i^*$, i.e. $F^{(j)}_{i^*} \geq 0.99$, where $F^{(j)}_i$ denotes the fidelity in iteration $i$ for the $j$th training run. For each run, we took as a figure of merit the maximum fidelity attained during training, $F^{(j)}_{\max} = \max_i F^{(j)}_i$. When multiple independent runs were performed, we report the average value of this same quantity:
\begin{equation}\label{eq:fidavg}
    F^{\mathrm{avg}}_{\max} 
    = 
    \frac{1}{M}
    \sum_{j=1}^M F^{(j)}_{\max}
\end{equation}
with $M$ the number of cases.

For both Hamiltonians considered, our numerical results show clear improvement when the auxiliary qubit is introduced midway through the training process ($i=3000$) with randomly initialized parameters, as shown in Fig.~\ref{fig:midrun}(b). The largest increase in average maximum fidelity is achieved when the auxiliary qubit is connected to all system qubits (configuration (a4)). Nevertheless, a substantial enhancement is already observed when it is coupled only to two previously distant qubits (configuration (a3)).
The improvement is closely linked to the redistribution of entanglement induced by the auxiliary qubit.
We discuss this mechanism in Section~\ref{sec:analysis}, but first  we identify the operational conditions required to trigger this entanglement-assisted learning enhancement, namely (i) the connectivity, (ii) the insertion timing, and (iii) the randomisation of the auxiliary system.

\subsection{Connectivity with the generator} \label{sec:connectivity}

To explore the first crucial aspect, we consider several connectivity configurations for each layer, as shown in Fig.~\ref{fig:midrun}(a): (a1) the auxiliary qubit interacts with a single qubit; (a2) the auxiliary qubit interacts with two neighbor qubits; (a3) the auxiliary qubit interacts with two distant qubits at the edges; and (a4) the auxiliary qubit interacts with all generator qubits.

By a dimensional expressivity analysis of the generator \cite{Funcke2021dimensional}, we can infer the expressivity of each case. This is done by computing the Jacobian matrix of the circuit with respect to its parameters and examining its effective rank, allowing us to determine the number of independent circuit parameters and, hence, which circuit has fewer redundant parameters and a higher expressivity. By computing the Jacobian matrix of the circuit, we determine the dimension of the manifold of states that is accessible to the circuit. 
If the rank equals the number of parameters, each of the parameters independently expands the reachable state space; if not, some of the parameters are redundant because they do not change the output state in a new direction. 

Let us denote the circuit corresponding to the auxiliary qubit configuration (a$n$) with $N$ parameters $\boldsymbol{\theta}_N = (\theta_1, \dots, \theta_N)$ as $C_{n}(\boldsymbol{\theta}_N)$, where $n \in \{1,2,3,4\}$. We define the Jacobian of the circuit as
\begin{equation}
J_n(\boldsymbol{\theta}_N) = \begin{pmatrix}
\frac{\partial C_n}{\partial \theta_1} &
\frac{\partial C_n}{\partial \theta_2} &
\cdots &
\frac{\partial C_n}{\partial \theta_N}
\end{pmatrix}.
\label{eq:jacobian}
\end{equation}
The expressivity of $C_n$ (normally given by the full dynamical Lie algebra) is then approximated as $\mathcal{E}_n \defeq \langle \operatorname{rank} \left(J_n(\boldsymbol{\theta}_N)\right)\rangle_{\boldsymbol{\theta}_N}$. We apply this analysis to the circuit configurations given in Fig.~\ref{fig:midrun}(a) and observe numerically that for a finite depth of three layers and non-zero parameters, the expressivities manifest a complete hierarchy. By randomly sampling parameters and computing the rank of the explicit Jacobian in Eq.~\eqref{eq:jacobian}, we empirically find $\mathcal{E}_4 > \mathcal{E}_3 > \mathcal{E}_2 > \mathcal{E}_1$ for intermediate rank tolerances (see Appendix~\ref{app:jacobian} for more details).

The dominance of $\mathcal{E}_3$ over $\mathcal{E}_2$ is a consequence of the circuit architecture topology, given that both configurations possess the exact same parameter count. While (a2) configuration creates redundant pathways between already adjacent qubits, (a3) configuration bridges distant qubits that are not linked by any gates. This shrinks the effective graph diameter and favours a more homogeneous entanglement distribution in the system, as will be illustrated in Section~\ref{sec:analysis} for configuration~(a4).

\subsection{Timing for the insertion of an auxiliary qubit}

The timing at which the auxiliary qubit is introduced plays a crucial role in determining the effectiveness of the learning protocol. The benefit of enlarging the Hilbert space is not uniform throughout the optimization process; it becomes particularly significant when the training dynamics exhibit clear signs of stagnation. Thus, we train the system for 3000 iterations, and then insert the auxiliary qubit in the way shown in Fig.~\ref{fig:midrun}(a), where the parameters of the auxiliary qubit are randomly initialized. Afterwards, the training continues for 3000 more iterations (6000 in total). For each configuration (a1-a4), we repeat the procedure for 100 cases and take the average fidelity for each of the tested configurations. As displayed in  Fig.~\ref{fig:midrun}(b1,b2), the improvement in learning performance is consistent with the hierarchy obtained previously, $\mathcal{E}_4 > \mathcal{E}_3 > \mathcal{E}_2 > \mathcal{E}_1$.  As will be formally analysed in Section~\ref{sec:analysis}, the layout of the auxiliary qubit connections dictates its effectiveness in escaping stagnation.

\begin{figure}[tb]
\centering
    \includegraphics[width=0.9\linewidth]{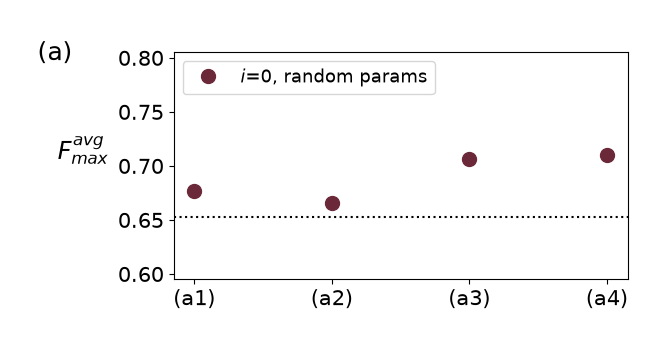}
    \\
    \vspace{-10px}
    \includegraphics[width=0.9\linewidth]{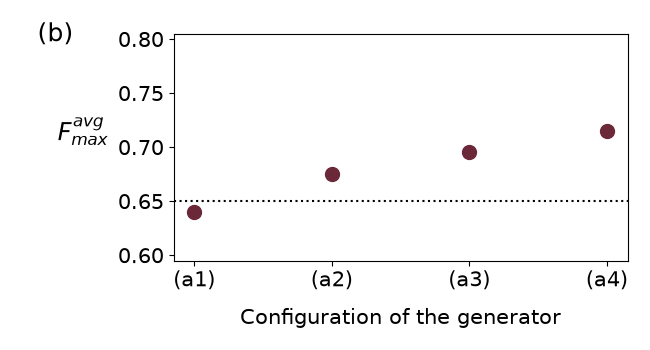}
    \caption{
    Learning unitary evolution, $U_T=e^{-iH_T}$, with an auxiliary qubit integrated into the generator from the beginning and trained for 6000 iterations, for comparison with Fig.~\ref{fig:midrun}(b1,b2).
    The dotted horizontal line indicates the reference fidelity achieved in the absence of an auxiliary qubit. In (a,b), we target $H_T=\sigma_z^1\sigma_z^2\sigma_z^3$ and $H_T=\sigma_z^1\sigma_z^2\sigma_z^3+\sigma_x^1\sigma_x^2\sigma_x^3/4$, respectively.
    }
    \label{fig:from_start}
\end{figure}

It is natural to ask whether the observed improvements stem solely from the enhanced generator expressivity provided by the additional auxiliary qubit parameters. To isolate this effect, we evaluate the training protocol with the auxiliary qubit present from the beginning.
In this setting, the generator operates in an enlarged Hilbert space throughout the entire training. While this increases the theoretical expressivity of the circuit, our numerical results indicate that such an early introduction does not translate into equally significant gains in performance, as shown in Fig.~\ref{fig:from_start}. As a general trend, after benchmarking our method on several non-commuting Hamiltonians, we observe that the fidelity gain remains considerably smaller than that obtained by adding the auxiliary qubit mid-training, although gains are still observed in certain cases.
We attribute this behavior to the optimization landscape already being high-dimensional at initialization, with the additional degrees of freedom introduced by the auxiliary qubit providing the greatest benefit when they become available precisely at the stage of training where optimization difficulties arise, rather than being present from the beginning.

Overall, these results highlight that the timing of the auxiliary-qubit insertion is a key ingredient in achieving performance gains. Introducing it too early leads to smaller benefits, whereas adding it at an intermediate stage, after the onset of stagnation, provides a significant boost in learning efficiency by enabling the system to escape to otherwise inaccessible regions of the optimization landscape.

\subsection{Parameter randomization}
\label{app:Randomised_restart}

The last crucial aspect is randomizing the initial parameters of the auxiliary qubit to effectively use the expanded optimization landscape. 
This randomization produces a ``kick'' to the system, and to study this effect we remove this kick and rather introduce an auxiliary qubit without interfering with the rest. 
The brown squares in Fig.~\ref{fig:midrun}(b1,b2) show the fidelities averaged over 100 attempts when the parameters for the auxiliary qubit are set to zero, and they do not vary from the value in the absence of the auxiliary qubit.

However, one may ask whether the randomization on the parameters alone is enough and it is not necessary to expand the dimension of the system with an auxiliary qubit. To answer this question, we study another approach where we simply randomize some of the parameters of the generator without adding additional auxiliary qubits throughout the training.
Specifically, in the middle of training (at 3000 iterations), we choose some of the parameters randomly given some ratio and replace their values with random numbers. Figure~\ref{fig:aveFvsrandom} shows the averaged maximum fidelities for different ratios of randomization. No large deviation is observed from the reference case. 
It is clear that a random kick alone is insufficient; the gain needs the kick to act in the Hilbert space enlarged by the auxiliary qubit.

This approach is similar to ``warm start'' heuristics~\cite{Egger2021Warm, kashif24}, where initial points are strategically selected to reach high fidelity. In our case, we select the parameters randomly for the second 3000 iteration routine. 
If the parameters to be changed are replaced more informatively, even better results could be obtained, but it is out of the scope of this work.

\begin{figure}[tb]
    \centering
    \includegraphics[width=0.85\linewidth]{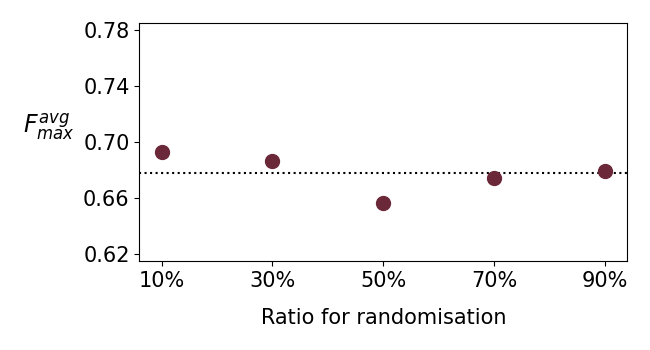}
    \caption{
    Randomized restart approach when targeting $e^{-i\sigma_z^1\sigma_z^2\sigma_z^3}$. The maximum reached fidelity averaged over 100 attempts is plotted for different parameter randomization ratios. The dotted lines represent the reference fidelity achieved without randomization.
    }
    \label{fig:aveFvsrandom}
\end{figure}

\section{Mechanism of the entanglement-assisted learning} \label{sec:analysis}

Up to here, we have discussed our strategy of entanglement-assisted learning. 
In this section, we will give more details of how the advantages are obtained. Specifically, we reveal the analytical form of the target dynamics to be learnt, and discuss the entanglement properties of the dynamics. 
Then, we perform the analysis of the gradients of the cost function to see how the auxiliary-qubit insertion changes the gradients.
Our findings illustrate the role of entanglement during the plateau (stagnation) regime of the training dynamics and support the view that the redistribution of entanglement induced by the auxiliary qubit is the key mechanism underlying the ability of the auxiliary-assisted approach to overcome stagnation and achieve improved performance.

\subsection{Entanglement enhancement} \label{sec:entanglement}

To gain some intuition on the learning process, 
we compute, for the target Hamiltonian,  $H=\sigma_z^1\sigma_z^2\sigma_z^3$,  its corresponding unitary dynamics,  $(\mathbb{I}_A\otimes e^{-i\sigma_z^1\sigma_z^2\sigma_z^3}) \ket{\Omega}$, where  $\ket{\Omega}$ is the input Choi state for $N=3$ (see Eq.~\eqref{eq:bipartite}). Notice that a proper reordering of the Choi state leads to the following correspondence:
\begin{equation}
\ket{\Omega}= \frac{1}{\sqrt{d}}\sum_i \ket{i}_A\ket{i}_B=\bigotimes_{k=1}^3 \ket{\Phi^+}_{A_k B_k}
\end{equation}
with $A_k,B_k$ being the indices for the $k$th qubit in $A,B$ respectively, and  $\ket{\Phi^+} = \left(\ket{00} + \ket{11}\right)/\sqrt{2}$.
Using that $e^{-i\sigma_z^1\sigma_z^2\sigma_z^3 t} 
    =(\cos t) \mathbb{I}_B - i (\sin t)
    \sigma_z^1\sigma_z^2\sigma_z^3$, the target dynamics is given by
\begin{align} \label{eq:target_dynamics}
    \left(
    \mathbb{I}_A\otimes 
    e^{-i\sigma_z^1\sigma_z^2\sigma_z^3 t}
    \right) \ket{\Omega}
    &=
    \left(\cos t\right)
    \bigotimes_{k=1}^3
     \ket{\Phi^+}_{A_k B_k}
     \nonumber\\
     &\quad
     - \left(i\sin t\right)
     \bigotimes_{k=1}^3
      \ket{\Phi^-}_{A_k B_k}
\end{align}
with $t=1$ and $\ket{\Phi^-} = \left(\ket{00} - \ket{11}\right)/\sqrt{2}$. 
Notice that by treating each pair $(A_k,B_k)$ as a single party, the above state corresponds to a GHZ-like state, and hence, to genuine multipartite entanglement (GME) with respect to these three parties, $k\in\{1,2,3\}$. The dynamics thus turns the input state $\ket{\Omega}$, a product of three Bell pairs (a maximally entangled state across the bipartition $A|B$ but also a product state with respect to the three $k$ parties), into a GME state.

\begin{figure*}
\centering
     \includegraphics[width=0.95\linewidth]{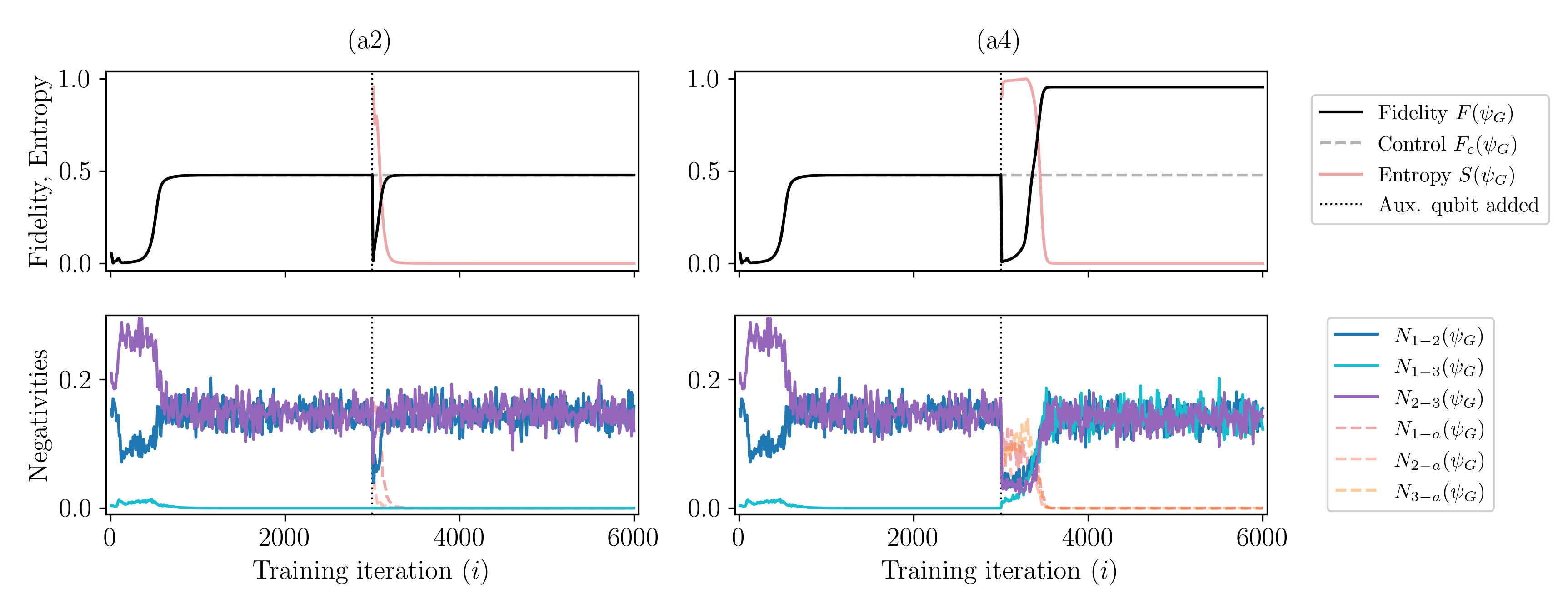}
     \caption{Typical training evolution of a QGAN when learning the unit time evolution of $H=\sigma_z^1 \sigma_z^2 \sigma_z^3$ under the (a2) and (a4) configurations. The top panels display the fidelity $F$, Eqs.~\eqref{eq:F_choi}, ~\eqref{eq:Choi_ancillary}, and the entanglement entropy $S(\rho_a)$, Eq.~\eqref{eq:entropy}. The control case without an auxiliary qubit is shown in semi-transparent dashed lines. The bottom panels display the bipartite negativity $\mathcal{N}$, Eq.~\eqref{eq:negativity}, tracking both the internal system pairs (solid lines) and the pairs with the auxiliary qubit (dashed lines).}
    \label{fig:midrun_live}
\end{figure*}

With this in mind, we study the generation of entanglement in the generator during the training. To this aim, we compute the von Neumann entropy (vNE) between the auxiliary system and the generator circuit qubits, and the negativity between all pairs of qubits in the extended generator. The  former is defined as
\begin{align} \label{eq:entropy}
    S(\rho_a)
    =
    -\mathrm{Tr}[\rho_a\log_2\rho_a]
\end{align}
Because of its simplicity, we use it to quantify the entanglement between the auxiliary qubit, given by the reduced density matrix $\rho_a$, and the rest of the generator. The negativity is defined for any two-qubit pair of the generator, $\rho_{ij}$, as
\begin{equation}
    \mathcal{N}(\rho_{ij}) 
    = 
    \frac{||\rho_{ij}^{T_i}||_1 - 1}{2},
    \label{eq:negativity}
\end{equation} 
where $\rho_{ij}^{T_i}$ is the partial transpose with respect to subsystem $i$ and $||\cdot||_1$ denotes the trace norm.
While the vNE can be directly evaluated from Eq.~\eqref{eq:entropy},  the negativity between pairs of qubits in the generator cannot be computed by tracing out the (irrelevant) degrees of freedom of the maximally entangled input training Choi state, Eq.~\eqref{eq:bipartite_extra}.
Instead, we benchmark the negativities across an ensemble of Haar-random product states decoupled from the auxiliary qubit, entering the generator at every iteration. See  Appendix~\ref{app:negativity} for more details.

In Fig.~\ref{fig:midrun_live}, we present the entanglement dynamics for a representative training trajectory in which the auxiliary qubit is introduced halfway through the optimization process ($i=3000$). Two distinct auxiliary system coupling configurations are considered. The top panels display the fidelity, $F$, between the target and generated states, Eqs.~\eqref{eq:F_choi} and \eqref{eq:Choi_ancillary}, together with the von Neumann entropy of the auxiliary qubit, $S(\rho_a)$.
Prior to the introduction of the auxiliary qubit, the dynamical behaviour of both configurations is, by construction, identical. During this stage, the output state exhibits a pronounced asymmetry in the distribution of correlations among the generator qubits, with entanglement being concentrated in a restricted subset of qubit pairs.

The coupling to the auxiliary qubit is marked by the sharp increase of $S(\rho_a)$, signalling the onset of correlations between the auxiliary and generator subsystems. Despite this common initial behaviour, the subsequent evolution differs significantly for the two coupling configurations. In configuration~(a2), the correlations established with the auxiliary qubit have only a marginal impact on the internal entanglement structure of the generator circuit, and the overall dynamical evolution remains qualitatively unchanged. Consequently, the optimization process continues to exhibit the same stagnation behavior observed prior to the coupling. In contrast, configuration~(a4) induces a qualitative reorganization of the pattern of pairwise correlations. The initial asymmetry is progressively suppressed, and entanglement becomes more evenly distributed across all qubits. This redistribution of quantum correlations is accompanied by a revival of the optimization dynamics, enabling the fidelity to overcome the previously observed plateau and converge towards unity. These results suggest that the performance enhancement is not solely associated with the generation of additional entanglement, but rather with its redistribution across the enlarged Hilbert space, thereby improving the trainability of the variational circuit.

Finally, given the GME nature of the target state, the fidelity $F$, defined in Eq.~\eqref{eq:Choi_ancillary}, can be employed as a witness of genuine multipartite entanglement. In this context, such a witness provides a natural and operational benchmark for assessing the performance of the learning protocol.
To establish the corresponding threshold, we determine the maximum overlap between the target state associated with the dynamics of Eq.~\eqref{eq:target_dynamics} and the set of biseparable states. As can be readily inferred from Eq.~\eqref{eq:target_dynamics}, the Schmidt coefficients across any bipartition of the three parties, namely $1|23$, $2|13$, and $3|12$, are given by $\sin(1)$ and $\cos(1)$. Therefore, the largest possible fidelity between the target state and any biseparable state is equal to the square of the largest Schmidt coefficient, namely $F_{\mathrm{bisep}}^{\max}=\sin^2(1)$. Consequently, any state achieving $F>\sin^2(1)$ is certified to be GME. This threshold will therefore serve as a quantitative indicator of the ability of the learning protocol to reproduce not only the target dynamics but also its multipartite entanglement structure. As observed in the example given in the right panel of Fig.~\ref{fig:midrun_live}, the fidelity crosses the bound $\sin^2(1)\approx 0.7081$ around $i=3500$, and the auxiliary-assisted generator succeeds in generating GME.

\subsection{Gradient statistics} \label{sec:landscape}

We conclude by analyzing the impact of the auxiliary qubit on the optimization landscape through a statistical study of the cost-function gradients. Specifically, we computed the partial derivatives of the Wasserstein cost function with respect to each variational parameter associated with the system of interest, and to obtain reliable statistics, we repeated the procedure over 1804 training runs that had not converged after the initial 3000 iterations without auxiliary qubit.
Our results are summarized in Fig.~\ref{fig:landscape}, which shows the gradient distributions at two representative stages of the optimization process: immediately at the insertion of the auxiliary qubit ($i=3000$) and 500 iterations later ($i=3500$), after the auxiliary-assisted training has had time to impact the learning. For a fair comparison with the non-auxiliary architecture, only the gradients corresponding to parameters of the original system are included in the analysis.

At the moment the auxiliary qubit is introduced, the gradient amplitudes exhibit a significant increase, indicating an immediate enhancement of the optimization signal, but the increase is similar for configurations~(a2), (a3), and (a4).  However, at a later stage ($i=3500$), the distributions become more structured, and their overall amplitudes are reduced, reflecting the progress of the training process. Notable differences now appear among the auxiliary-assisted architectures. In particular, the gradient distribution of configuration~(a4) contracts much less than the others and retains the largest gradient magnitudes, whereas configuration~(a1) retains the smallest. This observation suggests that the (a4) connectivity is more effective at sustaining informative gradients throughout the optimization, consistent with the superior learning performance observed in the previous analyses.
By contrast, no clear distinction can be established between configurations (a2) and (a3) at this stage of the training from the gradient histograms alone. Overall, these results indicate that the benefits of auxiliary-assisted learning are linked to an enhanced gradient quality and trainability. Nevertheless, the magnitude of the gradients provides only part of the explanation. Their geometrical properties, such as directionality and correlations in parameter space, likely also play an important role in determining the entanglement redistribution they can induce and, consequently, the optimization performance. A detailed investigation of these aspects is left for future work.

\begin{figure}[tb]
\centering
     \includegraphics[width=0.87\linewidth]{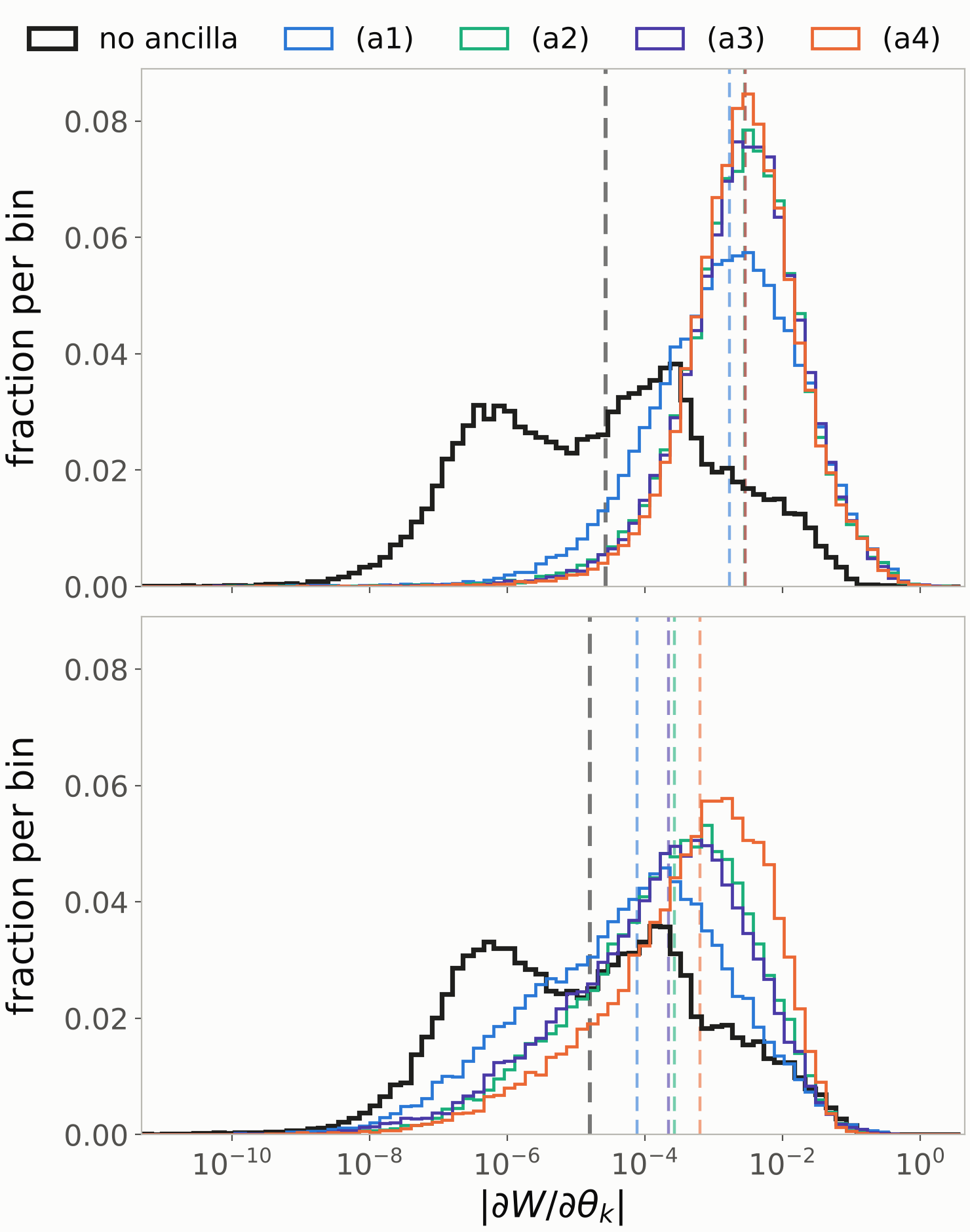}
     \caption{
     Normalised histograms for the amplitudes of the gradients when targeting $e^{-i\sigma_z^1\sigma_z^2\sigma_z^3}$ at 3000th (top) and 3500th iteration of training (bottom). 
     The vertical dashed lines show the medians of the histograms.
     Notice that the horizontal axis is on a logarithmic scale.
     }
     \label{fig:landscape}
\end{figure}

\section{Summary and Open Problems}\label{sec:summ}

A major obstacle in variational quantum algorithms is the emergence of learning plateaus and local minima in the optimization landscape, which severely hinders efficient training. These are particularly problematic when learning complex target unitaries using restricted generator models,
whose optimization landscapes are prone to suboptimal local minima.
Motivated by the need for practical strategies that alleviate stagnation without substantially increasing circuit depth, we have introduced an entanglement-assisted learning protocol that dynamically enlarges the optimization landscape while preserving a practical implementation.

In this work, we have shown that introducing an auxiliary qubit can significantly enhance learning performance when learning a target unitary $U_T$ with a simpler generator $U_G$. Three features of the auxiliary qubit are critical for this improvement: its capacity to entangle originally distant qubits, the random initialization of its parameters, and its insertion at an intermediate epoch of the training process rather than at the start. 
To identify what distinguishes the beneficial configurations, we have also tracked the growth of entanglement during training via the entanglement entropy and the pairwise negativities and shown that, in beneficial configurations, the auxiliary qubit helps the entanglement between distant qubits develop, which leads to higher average fidelity between the generator and the target at the end of training.

We note that increasing the complexity of the Hamiltonian, e.g., adding more non-commuting terms or four-body interactions and beyond, will demand only minimal modifications to the layer structure, preserving the generality of the approach and requiring minimal Hamiltonian-specific circuit engineering. 
In addition, the reliance of the protocol on random initialization and timed auxiliary-qubit insertion points toward more general strategies for dynamically varying the ansatz during training. It would be interesting to compare our strategy with alternatives such as increasing the circuit depth, adaptive circuit growth, and layer-wise training~\cite{Grimsley19, Tang21}. However, we emphasise that our goal has been to keep the circuit architecture as simple as possible; e.g., we employ only three non-adaptive training layers and add a single auxiliary qubit, while still achieving a substantial improvement in the performance of QGANs, without resorting to sophisticated methods for mitigating barren plateaus or training stagnation.
This idea could be integrated into noise-aware training \cite{Wang22}, and may be relevant also in variational quantum simulation \cite{Yuan19} and quantum control \cite{deKeijzer23}.

Several important questions remain open. First, how the benefits of auxiliary-qubit assistance scale with system size, circuit depth, and noise in realistic hardware implementations needs further exploration. Second, the optimal timing and number of auxiliary qubits required for different learning tasks have yet to be systematically characterized. Third, a deeper theoretical understanding of how auxiliary-qubit-induced entanglement reshapes the cost-function landscape and affects gradient statistics would be valuable. 
Addressing these questions will be crucial for the practical implementation of auxiliary-assisted strategies and for establishing their potential as scalable solutions to learning plateaus in quantum machine learning.

\section{Acknowledgments} We acknowledge stimulating discussions with María García Díaz, Mohammad Mehboudi, Fernando Vilariño, Paula García-Molina, Oriol Balló and Honey Verma Lal Rani. We also thank the anonymous referees for very insightful remarks.
A.U. and G.A.L. contributed equally to this work.
A.U. is financially supported by JSPS Overseas Research Fellowships. G.A.L. acknowledges support from the program of predoctoral funds FI-STEP (2025-3XJ3GQYK8), from Departament de Recerca i Universitats of Generalitat de Catalunya, and the co-financing by Fons Social Europeu Plus. SSB is supported by the Beatriu de Pinós programme of the Ministry of Research and Universities of the Government of Catalonia (grant number 2024 BP 00267).
We acknowledge support from MICINN grant PID2022-139099NBI00, with the support of FEDER funds, the Spanish Government with funding from European Union NextGenerationEU (PRTR-C17.I1), the Generalitat de Catalunya, the Ministry for Digital Transformation and of Civil Service of the Spanish Government through the QUANTUM ENIA project -Quantum Spain Project- through the Recovery, Transformation and Resilience Plan NextGeneration EU within the framework of the Digital Spain 2026.

\section{Data availability}
The code used for this work is available online~\cite{code, code_pennylane}.

\section{AI use disclosure}\label{sec:ai}

Gemini 3.8 Flash and Claude Opus 5.5 were used for small corrections to our main simulation code, to check the consistency in the derivations and improve the grammar and the lexicon in some written sections.
The more computationally demanding training runs were performed using a reimplementation of our QGAN training framework in PennyLane~\cite{Pennylane, code_pennylane}, developed with the assistance of Claude Opus 5.5 under the authors' direction. The authors take full responsibility for the content of this article.\\

\nocite{Egger2021Warm, kashif24}
\bibliography{bib.bib}

\clearpage
\appendix

\section{Quantum Wasserstein Semimetric and Optimization Implementation}\label{app:wasserstein}

In this Appendix, we provide the theoretical formulation and operational implementation of the quantum Wasserstein semimetric used as the adversarial cost function $\mathcal{L}$ in Section~\ref{sec:unitarylearning} by following Ref.~\cite{Chakrabarti2019Quantum}.

\subsection{Primal Formulation and the SWAP Cost}

For a more robust and scalable implementation of a quantum generative model, we formulate the optimization using the quantum Wasserstein semimetric of Ref.~\cite{Chakrabarti2019Quantum}. In the primal formulation, the distance between the target Choi state $\rho_T = |\Psi_T\rangle\langle\Psi_T|$ and a fixed generated state $\rho_G = |\Psi_G\rangle\langle\Psi_G|$ is defined via optimal transport over bipartite coupling states $\pi \in \mathcal{D}(\mathcal{H}_T \otimes \mathcal{H}_G)$:
\begin{equation}
   \mathcal{W} (\rho_G, \rho_T)= \min_{\pi} \mathrm{Tr}(\pi C) \quad \text{s.t.} 
   \quad 
   \begin{cases}
        \mathrm{Tr}_{\mathcal{H}_G}(\pi) = \rho_T \\
        \mathrm{Tr}_{\mathcal{H}_T}(\pi) = \rho_G,
    \end{cases}
\end{equation}
where $C$ is a positive semi-definite cost operator on the joint Hilbert space $\mathcal{H}_T \otimes \mathcal{H}_G$ and is given by
\begin{equation}
    C = \frac{1}{2}\left( \mathbb{I} - \mathrm{SWAP} \right),
\end{equation}
where $\mathrm{SWAP}$ exchanges the entire target and generator registers. Substituting this cost matrix into the objective yields $\mathrm{Tr}(\pi C) = \frac{1}{2} - \frac{1}{2}\mathrm{Tr}(\pi \, \mathrm{SWAP})$. 
Minimizing this cost corresponds to searching for a joint bipartite state $\pi$ that maximizes the SWAP overlap between its two subsystems while strictly maintaining the physical marginals $\rho_T$ and $\rho_G$, which for pure states gives $\mathcal{W}(\rho_G,\rho_T)=\frac{1}{2}\left(1-|\braket{\Psi_G}{\Psi_T}|^2\right)=\frac{1}{2}(1-F)$: since a bipartite state with a pure marginal is necessarily a product state, the only coupling compatible with both marginals is $\pi=\rho_T\otimes\rho_G$, and $\mathrm{Tr}[(\rho_T\otimes\rho_G)\,\mathrm{SWAP}]=\mathrm{Tr}(\rho_T\rho_G)=|\braket{\Psi_G}{\Psi_T}|^2$.
Finally, the QGAN objective is to minimize the following distance with respect to the generator parameters,
\begin{equation}
\min_G \mathcal{W}(\rho_G, \rho_T).
\end{equation}
When the generated state matches the target ($\rho_G = \rho_T$), the optimal coupling $\pi$ achieves maximal swap overlap, $\mathrm{Tr}(\pi \, \mathrm{SWAP})=1$, yielding $\mathcal{W}(\rho_G, \rho_T)= 0$.

\subsection{Dual Minimax Formulation and Observable Constraints}

Direct optimization over both the joint bipartite couplings $\pi$ and the generated state $\rho_G$ becomes quickly intractable due to the size of the parameter space. However, this primal minimization can be reformulated into the following equivalent dual min-max game over constrained local observables,
\begin{align}
    \min_G \max_{D_\beta, D_\alpha} \quad
    & 
    \mathrm{Tr}(\rho_T D_\beta) - \mathrm{Tr}(\rho_G D_\alpha)
    \\
    \text{s.t.} \quad
    & 
    D_\beta \otimes \mathbb{I}_G - \mathbb{I}_T \otimes D_\alpha \preceq C, \nonumber
\end{align}
where $D_\beta \in \mathcal{B}(\mathcal{H}_T)$ and $D_\alpha \in \mathcal{B}(\mathcal{H}_G)$ are Hermitian measurement operators parameterized by the discriminator. In this dual picture, the discriminator optimizes two local observables whose expectation values distinguish $\rho_T$ and $\rho_G$, subject to the joint constraint imposed by the cost operator $C$.

\subsection{Entropic Regularization and Differentiability}

To enforce the operator inequality constraint $D_\beta \otimes \mathbb{I}_G - \mathbb{I}_T \otimes D_\alpha \preceq C$ without computationally expensive semidefinite programming at each training epoch, we introduce the following constraint in an exponent of the cost as a smooth entropic penalty with regularization parameter $\lambda > 0$,
\begin{equation}
    \mathcal{L}_W 
    = 
    \mathrm{Tr}(\rho_T D_\beta) - \mathrm{Tr}(\rho_G D_\alpha) - \mathrm{Tr}\Big[ (\rho_T \otimes \rho_G) \xi_R \Big],
    \label{eq:reg_wloss}
\end{equation}
where the penalty operator $\xi_R$ is given by the matrix exponential,
\begin{equation}
    \xi_R 
    = 
    \frac{\lambda}{e} \exp\left( \frac{D_\beta \otimes \mathbb{I}_G - \mathbb{I}_T \otimes D_\alpha - C}{\lambda} \right).
\end{equation}
To facilitate efficient measurement decomposition on quantum hardware, we approximate the above exponential matrix and decouple it into subsystem terms,
\begin{equation}
    \xi_R 
    \approx 
    \frac{\lambda}{e} e^{-\frac{C}{2\lambda}} \left( e^{\frac{D_\beta}{\lambda}} \otimes e^{-\frac{D_\alpha}{\lambda}} \right) e^{-\frac{C}{2\lambda}}.
\end{equation}
The error is dependent on $\lambda$, and we choose the value empirically.

\subsection{Analytical Parameter Gradients}

We parameterize the discriminator observables using standard Pauli operator expansions,
\begin{equation}
    D_\beta 
    = 
    \sum_l \beta_l B_l, \quad D_\alpha = \sum_k \alpha_k A_k,
\end{equation}
where $A_k, B_l$ are Pauli strings acting on the individual target and generator spaces. The gradients of the regularized Wasserstein loss $\mathcal{L}_W$ with respect to the discriminator parameters evaluate to
\begin{subequations}
\begin{align}
    \frac{\partial \mathcal{L}_W}{\partial \beta_l}
    &\approx 
    \mathrm{Tr}(\rho_T B_l) - \frac{1}{\lambda} \mathrm{Tr}\Big[ (\rho_T \otimes \rho_G) (B_l \otimes \mathbb{I}_G) \xi_R \Big],
    \\
    \frac{\partial \mathcal{L}_W}{\partial \alpha_k}
    &\approx
    -\mathrm{Tr}(\rho_G A_k) + \frac{1}{\lambda} \mathrm{Tr}\Big[ (\rho_T \otimes \rho_G) (\mathbb{I}_T \otimes A_k) \xi_R \Big].
\end{align}
\end{subequations}

For the generator parameters $\theta_j$, the differentiation of Eq.~\eqref{eq:reg_wloss} yields:
\begin{equation}
    \frac{\partial \mathcal{L}_W}{\partial \theta_j} 
    = 
    -\mathrm{Tr}\left(\frac{\partial \rho_G}{\partial \theta_j}  D_\alpha\right) -  \mathrm{Tr}\left[ \Big(\rho_T \otimes \frac{\partial \rho_G}{\partial \theta_j}\Big) \xi_R \right],
\end{equation} 
where each expectation value is evaluated on the state $\rho_G(\boldsymbol{\theta}_N)$ using standard parameter-shift rules.

\subsection{Momentum-Based Parameter Optimization}

At each training step, the parameters of the discriminator $\boldsymbol{\phi} = (\boldsymbol{\alpha}, \boldsymbol{\beta})$ and generator $\boldsymbol{\theta}_N$ are updated alternately using a classical momentum-based optimizer. For the discriminator parameters $\boldsymbol{\phi}$, the update proceeds in two sequential steps:
\begin{subequations}
\begin{align}
    \mathbf{v}_{D}
    &\leftarrow 
    \beta_D \mathbf{v}_{D} + (1 - \beta_D) \nabla_{\boldsymbol{\phi}} \mathcal{L}_W, 
    \\
    \boldsymbol{\phi}
    &\leftarrow 
    \boldsymbol{\phi} + \eta_D \mathbf{v}_{D},
\end{align}
\end{subequations}
where $\eta_D$ is the discriminator learning rate, $\beta_D \in [0, 1)$ is the momentum hyperparameter, and $\mathbf{v}_D$ is the velocity vector initialized to $\mathbf{0}$. The generator parameters $\boldsymbol{\theta}_N$ follow the identical update rule with velocity $\mathbf{v}_G$, learning rate $\eta_G$, and momentum $\beta_G$, moving along the descent direction $\boldsymbol{\theta}_N \leftarrow \boldsymbol{\theta}_N - \eta_G \mathbf{v}_G$.

\section{Dimensional Expressivity and Jacobian Rank}\label{app:jacobian}

To numerically evaluate the expressivity $\mathcal{E}_n$ discussed in Section~\ref{sec:connectivity}, we compute the effective rank of the Jacobian matrix $J_n(\boldsymbol{\theta}_N)$. We describe the explicit calculation for $C_1$, which utilizes eleven parameters per layer, ordered as shown in Fig.~\ref{fig:midrun}(a). The calculations for the other circuits follow similarly. We define the base unitary rotations and the entangling operations for a single layer as:
\begin{widetext}
\begin{align}
U_{0}(\theta_1,...,\theta_8) &= \big( R_{z}(\theta_5)\, R_{x}(\theta_1) \big) \otimes \big( R_{z}(\theta_6)\, R_{x}(\theta_2) \big) \otimes \big( R_{z}(\theta_7)\, R_{x}(\theta_3) \big) \otimes \big( R_{z}(\theta_8)\, R_{x}(\theta_4) \big), \\
U_{\mathrm{ent}}(\theta_9,...,\theta_{11}) &= \big( \mathbb{I} \otimes \mathbb{I} \otimes R_{ZZ}(\theta_{11}) \big) \big( \mathbb{I} \otimes R_{ZZ}(\theta_{10}) \otimes \mathbb{I} \big) \big( R_{ZZ}(\theta_{9}) \otimes \mathbb{I} \otimes \mathbb{I} \big).
\end{align}   
\end{widetext}

The full circuit for three layers is $C_1(\boldsymbol{\theta}_N)=(U_{\mathrm{ent}}U_0)^3$ with independent parameters in each layer, where the rightmost factor acts first. Writing $C_1$ as the ordered product of its 33 gates, $C_1=g_{33}\cdots g_2\,g_1$, with $g_k=e^{-i\theta_k P_k/2}$ and $P_k$ the Pauli generator ($\sigma_x$, $\sigma_z$, or $\sigma_z\otimes\sigma_z$) of the $k$th gate, ordered as in $U_0$ and $U_{\mathrm{ent}}$ above, each derivative reads
\begin{equation}
    \partial_{\theta_k}C_1=-\frac{i}{2}\,\big(g_{33}\cdots g_{k+1}\big)\,P_k\,\big(g_k\cdots g_1\big).
\end{equation}
The Jacobian matrix $J_1(\boldsymbol{\theta}_N)$ is formed by concatenating the 33 vectors $\partial_{\theta_k}C_1\ket{0000}$. Then, to obtain its average rank numerically, we sample $10^3$ random parameter vectors $\boldsymbol{\theta}_N$ uniformly and evaluate them.

Crucially, we repeat this process for different relative tolerances of the rank computation. The relative tolerance acts as a threshold, discarding singular value contributions that are relatively smaller than the maximum singular value of $J_1(\boldsymbol{\theta}_N)$. We apply the identical algorithmic procedure to configurations $C_2$, $C_3$, and $C_4$, which possess 36, 36, and 39 parameters, respectively. For relative tolerances bounded between $10^{-3}$ and $0.5$, 
a hierarchy, $\mathcal{E}_4> \mathcal{E}_3> \mathcal{E}_2 > \mathcal{E}_1$, is observed as presented in Fig.~\ref{fig:expressivity}.
While it is predictable that $\mathcal{E}_4$ and $\mathcal{E}_1$ are the largest and smallest, respectively, it is interesting that $\mathcal{E}_3>\mathcal{E}_2$, i.e. the expressivity of the configuration~(a3) is larger than that of (a2).

\begin{figure}[tb]
    \centering
    \includegraphics[width=0.9\linewidth]{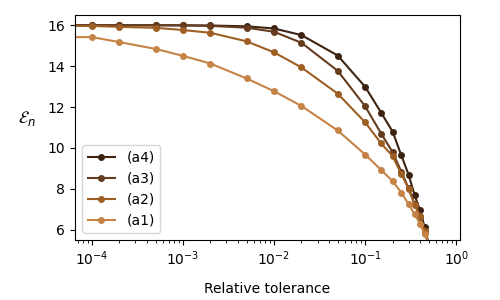}
    \caption{Circuit expressivity $\mathcal{E}_n$ given by its Jacobian rank, as a function of the relative tolerance in computing such rank for each auxiliary-qubit connectivity. Each data point was obtained by averaging over $10^3$ random parameter vectors $\boldsymbol{\theta}_N$.}
    \label{fig:expressivity}
\end{figure}

The dependence on the relative tolerance reflects the distribution of the singular values of the Jacobian matrix. Smaller tolerances capture marginal contributions where, statistically, all the cases are full rank, while larger tolerances isolate the dominant singular vectors, with all the cases converging to an effective single rank contribution. Therefore, the dependency on the intermediate relative tolerances indicates how information is distributed across the different dimensions for each configuration.\\

\section{Internal System Negativity Baseline and Ensemble Validation} \label{app:negativity}

In Section~\ref{sec:analysis}, we discuss the computation of pairwise negativities during training. These negativities cannot be computed by employing the training Choi state, Eq.~\eqref{eq:bipartite_extra}, because tracing out the reference register $A$ leaves the system in a maximally depolarized marginal $\rho_B = \mathbb{I}/8$, and the entanglement of the training Choi state between the regions $A,B$ drowns out the correlations inside the generator.
Instead, we sample the generator with an ensemble of $50$ Haar-random product states, $|\Psi_S\rangle = \bigotimes_{k=1}^3 |\phi_k\rangle$, initialized with the auxiliary qubit state being $|0\rangle_a$ at every iteration.

Evaluating the target unitary $U_T = e^{-i \sigma_z^1 \sigma_z^2 \sigma_z^3}$ across this ensemble yields an isotropic baseline of
\begin{equation}
    \langle \mathcal{N}(\rho_{ij}) \rangle = 0.1288 \pm 0.0153,
\end{equation}
with standard error of the mean $\mathrm{SEM} = \sigma / \sqrt{50} \approx 0.0153$ (population standard deviation $\sigma \approx 0.1079$). This establishes the stationary target baseline toward which internal system pairs relax upon successful optimization (the negativities between the auxiliary qubit and the system vanish to zero). Continuous Haar sampling is crucial here as deterministic sets in the computational basis are stationary eigenstates of $\sigma_z^1 \sigma_z^2 \sigma_z^3$, which leads to $\mathcal{N}=0$, whereas Haar-distributed states provide an unbiased Monte Carlo estimate of the generator's entangling power.

\section{Learning Unitary Evolution in a Four-qubit System} \label{sec:bigger_H}

In the main text, we focused on three-qubit systems for a target. Here, we extend the system size and study a four-qubit system to see whether our results are consistent in a larger system.
We target a 3-local Hamiltonian, $H=\sigma_z^1\sigma_z^2\sigma_z^3+ \sigma_z^2\sigma_z^3\sigma_z^4$, acting on  four qubits. 
As done in Section~\ref{sec:aqgan}, we start the training for QGAN for 3000 iterations and continue the training for another 3000 iterations after inserting an auxiliary qubit in a way depicted in Fig.~\ref{fig:midrun}(a1-4). The average fidelities over 100 cases are shown in Fig.~\ref{fig:larger}. 
When the auxiliary qubit connects with everyone, the most benefit is obtained, which is consistent with the three-qubit case. On the other hand, when the auxiliary qubit connects with the two qubits at the edges, the benefit is much less and comparable to the case without the auxiliary qubit. It indicates that the auxiliary qubit needs to interact with more qubits or differently to break the learning plateaus. 
It would be informative to track the growth of entanglement during the training as done in Section~\ref{sec:analysis}, and it could help one with the auxiliary-qubit configuration.

To compute this instance within an adversarial minimax loop, we implemented an accelerated execution backend using PennyLane~\cite{Pennylane} and PyTorch~\cite{paszke2019pytorch}. See our code~\cite{code_pennylane} for more details.\\

\begin{figure}[H]
\centering
     \includegraphics[width=0.9\linewidth]{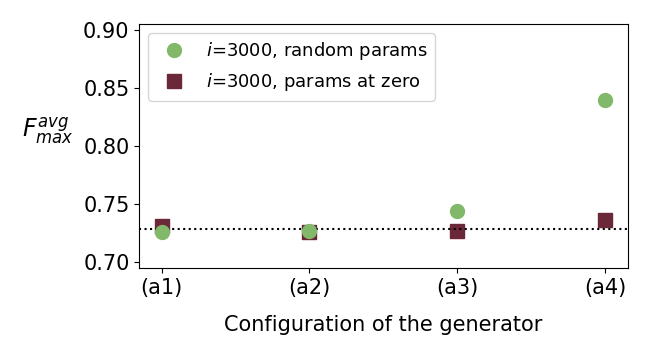}
     \caption{
     Entanglement-assisted learning when learning the time evolution, $e^{-i\left(\sigma_z^1\sigma_z^2\sigma_z^3+ \sigma_z^2\sigma_z^3\sigma_z^4\right)}$, acting on four qubits.
     The average maximum fidelities for different auxiliary-qubit configurations are plotted. 
     The dotted horizontal line indicates the reference fidelity achieved in the absence of an auxiliary qubit. 
     The green dots (brown squares) denote the average fidelities over 100 cases when the parameters of the auxiliary qubit are initialized randomly (set to zero).
     }
     \label{fig:larger}
\end{figure}

\end{document}